\documentclass{article}

\usepackage{latexsym}
\usepackage{amssymb}
\usepackage{amsfonts}
\usepackage{amstext}
\usepackage{multicol}
\usepackage{multirow}
\usepackage{enumerate}
\usepackage{amsmath}
\usepackage{graphicx, psfrag}
\usepackage{array}
\usepackage{mathrsfs}
\usepackage{url}
\usepackage{amsthm}

\newcommand{\eps}{\varepsilon}
\newcommand{\R}{\mathbb{R}}
\newcommand{\C}{\mathbb{C}}
\newcommand{\N}{\mathbb{N}}
\newcommand{\uinf}{0}
\newcommand{\D}{\partial}
\newcommand{\ham}{\bar{H}_l}
\renewcommand{\Im}{\text{Im }}
\renewcommand{\Re}{\text{Re}}
\newcommand{\etaa}{\eta_{1,\omega}}
\newcommand{\etab}{\eta_{2,\omega}}
\newcommand{\zetaa}{\zeta_{1,\omega}}
\newcommand{\zetab}{\zeta_{2,\omega}}
\newcommand{\gammaa}{\gamma_{1,\omega}}
\newcommand{\gammab}{\gamma_{1,\omega}}
\newcommand{\Gammaa}{\Gamma^a_{\omega}}
\newcommand{\Gammab}{\Gamma^b_{\omega}}
\newtheorem{thm}{Theorem}[section]
\newtheorem{prop}{Proposition}[section]
\newtheorem{cor}{Corollary}[section]
\numberwithin{equation}{section}
\author{Matthew Masarik\\
University of Michigan\\
\texttt{masarikm@umich.edu}}
\title{The Wave Equation in a General Spherically Symmetric Black Hole Geometry}
\date{\today}

\begin{document}
\maketitle
\begin{abstract}We consider the Cauchy problem for the wave equation in a general class of spherically symmetric black hole geometries. Under certain mild conditions on the far-field decay and the singularity, we show that there is a unique globally smooth solution to the Cauchy problem for the wave equation with data compactly supported away from the horizon that is compactly supported for all times and \emph{decays in $L^{\infty}_{\text{loc}}$ as $t$ tends to infinity}. We obtain as a corollary that in the geometry of black hole solutions of the SU(2) Einstein/Yang-Mills equations, solutions to the wave equation with compactly supported initial data decay as $t$ goes to infinity.\end{abstract}
\section{Introduction}The Cauchy problem for the wave equation in various black hole geometries is an active area of research with much effort devoted to showing decay of the solutions. In the case of the Schwarzschild metric, Kronthaler showed in \cite{thecauchyproblemforthewaveequationintheschwarzschildgeometry} that there exists a unique global solution to this problem and, moreover, the solution decays pointwise as $t \to \infty$. In \cite{onpointwisedecayoflinearwavesonaSchwarzschildblackholebackground} Donninger, Schlag, and Soffer obtain the specific decay \emph{rate} $t^{-3}$ for solutions; and in \cite{decayratesforsphericalscalarwavesintheSchwarzschildgeometry} Kronthaler obtains the same rate under the assumption that the data is spherically symmetric (i.e. for the first angular mode of the full solution), along with the additional result that if the data is momentarily static (i.e., $\left.\D_t\phi\right|_{(x,0)} = 0$), then the decay rate can be improved to $t^{-4}$ (again for the first angular mode). For the Kerr metric (without a smallness restriction to the angular momentum), Finster, Kamran, Smoller, and Yau showed decay of solutions of the wave equation in \cite{decayofsolutionsofthewaveequationinthekerrgeometry}. If the case of sufficiently small angular momentum, Dafermos and Rodnianski were able to demonstrate the uniform boundedness of solutions to the wave equation in \cite{aproofoftheuniformboundednessofsolutionstothewaveequationsonslowlyrotatingkerrbackgrounds}, and Andersson and Blue obtained decay \emph{rates} in \cite{hiddensymmetriesanddecayforthewaveequationonthekerrspacetime}, again for sufficiently small angular momentum.

Our goal in this paper is to generalize decay results in the Schwarzschild metric to a much more general class of spherically symmetric black hole geometries. We consider a metric given by
\begin{equation}ds^2 = g_{ij}dx^idx^j = -T^{-2}(r)dt^2 + K^2(r)dr^2 + r^2(d\theta^2 + \sin^2\theta d\phi^2),\label{metric}\end{equation}
where $r>0, 0\leq \theta\leq \pi, 0\leq \phi\leq 2\pi$. This clearly generalizes the Schwarzschild geometry, as can be seen by making the identifications
\begin{equation*}T(r) = \left(1 - \frac{2m}{r}\right)^{-\frac{1}{2}} = K(r).\end{equation*}
We must impose natural (i.e. physical) conditions on the coefficients $T,K$. To that end, we assume there is a singularity in $K$ at $r = r_0>0$ and near the singularity we have the following asymptotics: 
\begin{equation}T(r) = c_1(r - r_0)^{-\frac{1}{2}} + O(1) \text{ and } K(r) = c_2(r - r_0)^{-\frac{1}{2}} + O(1)\label{eventhorizonasymptotics1}\end{equation}
for some constants $c_1,c_2>0$, and 
\begin{equation}T'(r) = c_3(r-r_0)^{-\frac{3}{2}} + O(r-r_0)^{-\frac{1}{2}} \text{ and }K'(r) = c_4(r - r_0)^{-\frac{3}{2}} + O(r-r_0)^{-\frac{1}{2}}\label{eventhorizonasymptotics2}\end{equation}
for some nonzero constants $c_3,c_4$. We assume smoothness away from the horizon: $T,K \in C^{\infty}(r_0,\infty)$; and we assume that in the far-field, the metric asymptotically approaches Minkowski flat-space: 
\begin{equation}T(r) = 1 + O\left(\frac{1}{r}\right) \text{ and }K(r) = 1 + O\left(\frac{1}{r}\right)  \text{ as } r \text{ tends to infinity}.\label{farfieldasymptotics1}\end{equation}
We assume that for each $r \in (r_0,\infty)$, $T(r) \neq 0$ and $K(r)\neq 0$; and finally, we impose restrictions on the far-field decay: 
\begin{equation}\frac{T'(r)}{T(r)} + \frac{K'(r)}{K(r)} = O\left(\frac{1}{r^2}\right) \text{ as } r \text{ tends to infinity}.\label{farfieldasymptotics2}\end{equation}
We note that the Schwarzschild metric, the non-extreme Reissner-Nordstrom metric, and the metrics given by black hole solutions of the Einstein/Yang-Mills (EYM) equations (c.f. \cite{existenceofblackholesolutionsfortheeymequations}) satisfy these conditions. It is easy to see the first two cases, and we will show in Section~\ref{eym} that black hole solutions to the EYM equations satisfy these conditions. 
For the purposes of this paper, a geometry $\eqref{metric}$ satisfying the above conditions will be referred to as a \emph{spherically symmetric black hole} (SSBH). We note also that the work \cite{thecauchyproblemforthewaveequationintheschwarzschildgeometry} served as a model for solving this problem and we also rely on results therein in a few places.

\section{Preliminary Notions}We begin by recalling that, given a spacetime metric $g$, the scalar wave equation in that geometry is given by
\begin{equation} \square \zeta := g^{ij}\nabla_i\nabla_j \zeta = \frac{1}{\sqrt{-g}}\frac{\partial}{\partial x^i}\left(\sqrt{-g}\,g^{ij}\,\frac{\partial}{\partial x^j}\right)\zeta = 0,\label{defofwaveop}\end{equation} 
where $g^{ij}$ is the inverse of the metric $g_{ij}$ and $g = \det(g_{ij})$. So, in the geometry $\eqref{metric}$, the wave equation takes the form 
\begin{equation}\square \zeta= \left(-T^2 \partial_t^2 + \frac{1}{r^2}\D_r\left(\frac{r^2\partial_r}{K^2}\right) + \frac{T}{K^3}\partial_r \left(\frac{K}{T}\right)\partial_r + \frac{1}{r^2}\Delta_{S^2}\right)\zeta = 0,\label{waveequation1}\end{equation}
where $T=T(r),K= K(r),$ and $\Delta_{S^2} = \frac{\partial}{\partial(\cos\theta)}\left(\sin^2\theta \frac{\partial}{\partial(\cos \theta)}\right) + \frac{1}{\sin^2\theta}\frac{\partial^2}{\partial \phi^2}$ is the standard Laplacian on the sphere $S^2$. We introduce the coordinate $u = u(r)$ by 
\begin{equation}u(r) = -\int_{r}^{\infty} \frac{K(\alpha)T(\alpha)}{\alpha^2} d\alpha,\label{defofu}\end{equation}
which maps the interval $(r_0,\infty)$ to the interval $(-\infty,0)$. This is a simple consequence of the prescribed asymptotics. We note also that since $T,K$ are everywhere positive, $u$ is indeed a valid coordinate change and the inverse mapping $r = r(u)$ is well-defined. Then the wave equation $\eqref{waveequation1}$ on $\R \times (r_0, \infty) \times S^2$ is equivalent to 
\begin{equation}\left(-r^4 \partial_t^2 + \partial_u^2 + \frac{r^2}{T^2}\Delta_{S^2}\right)\psi = 0\label{waveequation2}\end{equation}
on $\R \times (-\infty,0) \times S^2$, where $T = T(r)$ and $r = r(u)$ (in what follows, we will frequently suppress the arguments of these functions). The Cauchy problem for the wave equation in the coordinates $(t,u,\theta,\phi)$ then reads\footnote{We use the compact form $(\psi,i\psi_t)$ for the data in what follows, since this is most convenient when we reformulate this as a Hamiltonian problem later.}
\begin{equation}\begin{cases} \left(-r^4(u) \partial_t^2 + \partial_u^2 + \frac{r^2(u)}{T^2\left(r(u)\right)}\Delta_{S^2}\right)\psi(t,u,\theta,\phi) = 0 \text{ on } \R \times (-\infty, \uinf) \times S^2,\\
(\psi, i\psi_t)(0,u,\theta,\phi) = \Psi_0(u,\theta,\phi) \in C_0^{\infty}\left((-\infty,\uinf) \times S^2\right)^2.\end{cases}\label{cp}\end{equation}

Let us settle first the question of existence and uniqueness for the problem $\eqref{cp}$.

\begin{thm} The Cauchy problem $\eqref{cp}$ in the geometry of an SSBH has a unique, smooth solution that exists for all times $t$. Furthermore, this solution is compactly supported in $(u,\theta,\phi)$ for each time $t$.\label{existenceanduniqueness}\end{thm}
\begin{proof}To prove the theorem, we wish to apply the theory of symmetric hyperbolic systems in section 5.3 of \cite{fritz} to the auxiliary PDE 
\begin{equation}\left(\D_t^2 - \D_s^2 - \frac{\Delta_{S^2}}{r^2 T^2} - \frac{1}{T^2K^2r}\left(\frac{T'}{T} + \frac{K'}{K}\right)\right)\xi(t,s,\theta,\phi) = 0,\label{auxpde}\end{equation}
where 
\begin{equation}s(u) = \int_{u(2r_0)}^u r^2(\alpha)d\alpha,\label{scoord}\end{equation}
$r = r(u(s))$, and the arguments of $T,T',K,K'$ are also $r(u(s))$. Note that we may consider $s$ as a function of $r$ by considering
\begin{align}s(u(r)) &= \int_{u(2r_0)}^{u(r)}r^2(\alpha)d\alpha\notag \\ &= \int_{2r_0}^r T(\alpha)K(\alpha)d\alpha.\label{s(r)}\end{align}
The PDE $\eqref{auxpde}$ is equivalent to $\eqref{waveequation2}$ upon making the change of coordinate $s = s(u)$ and letting $\xi = r\psi$. We consider this PDE because in this coordinate we will be able to prove that the solution of $\eqref{cp}$ is compactly supported for each $t$. In the $u$ variable, this is not obvious, and it is not clear how to prove it. Let us also note that $s(u)$ maps the interval $(-\infty,\uinf)$ monotonically onto $\R$. We will prove the theorem first for the Cauchy problem
\begin{equation}\begin{cases}\left(\D_t^2 - \D_s^2 - \frac{\Delta_{S^2}}{r^2 T^2} - \frac{1}{T^2K^2r}\left(\frac{T'}{T} + \frac{K'}{K}\right)\right)\xi = 0 \text{ on } \R \times \R \times S^2,\\
(\xi,i\xi_t)(0,s,\theta,\phi) = \Xi_0(s,\theta,\phi) \in C_0^{\infty}(\R \times S^2)^2,\end{cases}\label{cp4}\end{equation}
and then use this to obtain results about the Cauchy problem $\eqref{cp}$. We note here that the argument that follows is based on a similar argument in \cite{thecauchyproblemforthewaveequationintheschwarzschildgeometry}. To prove the theorem, we must work in local coordinates on $S^2$. So let us consider the chart $(U,(\theta,\phi))$ where $U$ is an open, relatively compact subset of $S^2$ and $(\theta,\phi)$ are well-defined on $\bar{U}$. Then, letting $\Gamma = (\xi_t, \xi_u, \partial_{(\cos\theta)}\xi, \partial_{\phi}\xi,\xi)^T$, we can cast the PDE in $\eqref{cp4}$ as a first-order system:
\begin{equation}A_0 \partial_t \Gamma + A_1 \partial_u \Gamma + A_2 \D_{(\cos\theta)} \Gamma + A_3 \D_\phi \Gamma +B\Gamma= 0,\label{symmhypsystem}\end{equation}
where the matrices $A_0, \dots, A_3,B$ are defined as follows:
\[ A_0 := \text{diag}\left(1, 1, \frac{\sin^2\theta}{r^2T^2}, \frac{1}{r^2T^2}\frac{1}{\sin^2 \theta},1\right), (A_1)_{12} = -1 = (A_1)_{21}, \]
\[(A_2)_{13} = -\frac{\sin^2\theta}{r^2T^2} = (A_2)_{31}, (A_3)_{14} = -\frac{1}{r^2T^2}\frac{1}{\sin^2\theta} = (A_3)_{41},\]
\[B_{13} = \frac{2\cos\theta}{r^2T^2}, B_{15} = -\frac{1}{T^2K^2r}\left(\frac{T'}{T} + \frac{K'}{K}\right), B_{51} = -1,\]
and all other matrix entries are zero. Upon multiplying this system by $T^2$, we obtain a symmetric hyperbolic system on $\R \times \R \times U$, since each $A_i$ is symmetric and $A_0$ is uniformly positive definite on this region. Further, since the initial data $\Xi_0$ has compact support, we can restrict the system to $\R \times V \times U$, where $V$ is open, relatively compact, and $\text{supp}\, \Xi_0(u,\theta,\phi) \subset V \subset \R$ for each $(\theta,\phi) \in S^2$. Since we can cover $S^2$ by finitely many such charts, the theory of symmetric hyperbolic systems guarantees the existence and uniqueness of a smooth solution $\xi$ of $\eqref{cp4}$ defined for all $t < \eps_1$ for some $\eps_1 > 0$. Moreover, this solution propagates with finite speed and thus there exists an $0 < \eps \leq \eps_1$ so that $\xi$ has compact support in $V\times S^2$ for all times $t\leq \eps$. Therefore, we can repeat this argument for the Cauchy problem with data $(\xi(\eps, u, \theta, \phi), i\xi_t(\eps, u, \theta, \phi))^T$ and obtain a unique, smooth solution $\xi$ defined for $t\leq 2\eps$ which has compact support in a possibly larger, though still open and relatively compact set for all times $t\leq 2\eps$. Repeating the argument yields a global solution $\xi$ of the Cauchy problem $\eqref{cp4}$ which is smooth, unique, and compactly supported for all times $t$.

Since we have already observed that a solution $\psi$ of $\eqref{cp}$ yields a solution of $\eqref{cp4}$ under the coordinate change $s = s(u)$ and the identification $\xi = r\psi$ and vice versa, the theorem follows. 
\end{proof}

We now observe that the Cauchy problem admits a conserved energy:

\begin{prop}A solution of the Cauchy problem $\eqref{cp}$ admits a conserved energy $E(\psi)$ given by
\begin{equation}E(\psi) = \int_0^{2\pi}\!\!\!\int_{-1}^1 \!\int_{-\infty}^{\uinf} \!\!r^4(u) \left(\psi_t\right)^2 + \left(\psi_u\right)^2 + \frac{r^2(u)}{T^2(r(u))}\left(\frac{1}{\sin^2\theta}\left(\D_{\phi}\psi\right)^2 + \sin^2\theta\left(\D_{(\cos\theta)}\psi\right)^2\right)dud(\cos\theta)d\phi;\label{energy}\end{equation}
i.e. $\frac{d}{dt}E(\psi) = 0$.\end{prop}
\begin{proof}We know that $\eqref{cp}$ admits a globally defined, smooth, unique solution which is compactly supported for all times $t$. Thus, the energy $E(\psi)$ is well-defined. Moreover, since $\psi$ solves $\eqref{cp}$, an easy calculation shows that $\frac{d}{dt}E(\psi) = 0$.
\end{proof}

Next, we wish to cast the Cauchy problem $\eqref{cp}$ as a first-order Hamiltonian system. To this end we define $\Psi := (\psi, i\psi_t)^T$; then $i\D_t \Psi = H\Psi$, where 
\[H = \left(\begin{array}{cc} 0 & 1 \\ A & 0\end{array}\right),\]
and $A = -\frac{1}{r^4}\D_u^2 - \frac{\Delta_{S^2}}{r^2 T^2}$. Therefore, the Cauchy problem $\eqref{cp}$ is equivalent to the problem
\begin{equation}\begin{cases} i\D_t \Psi = H\Psi \text{ on } \R \times (-\infty,\uinf) \times S^2, \\ \Psi(0,u,\theta,\phi) = \Phi_0(u,\theta,\phi) \in C_0^{\infty}\left((-\infty,\uinf) \times S^2\right)^2.\end{cases}\label{cp2}\end{equation}
Theorem \ref{existenceanduniqueness} implies that the problem $\eqref{cp2}$ has a unique, smooth solution $\Psi$ that is defined for all times $t$ and compactly supported for each $t$.

Let us next observe that the energy in $\eqref{energy}$ defines an inner product on the space $C_0^{\infty}\left((-\infty,\uinf)\times S^2\right)^2$: for $\Psi, \Gamma \in C_0^{\infty}\left((-\infty,\uinf)\times S^2\right)^2$ with $\Psi = (\psi_1,\psi_2)^T$ and $\Gamma = (\gamma_1, \gamma_2)^T$, we can define the scalar product $\langle \Psi, \Gamma\rangle$ by
\begin{equation}\int_0^{2\pi}\!\!\!\int_{-1}^1\!\int_{-\infty}^{\uinf} \!\!r^4 \psi_2 \overline{\gamma_2} + (\D_u\psi_1)\overline{(\D_u\Gamma_1)} + \frac{r^2}{T^2}\left(\frac{1}{\sin^2\theta}(\D_{\phi} \psi_1)\overline{(\D_{\phi}\gamma_1)} + \sin^2\theta(\D_{(\cos\theta)}\psi_1)\overline{(\D_{(\cos\theta)}\gamma_1)}\right)dud(\cos\theta)d\phi.\label{innerproduct1}\end{equation}
We next show that with respect to this inner product, $H$ is symmetric on the domain $C_0^{\infty}\left((-\infty,\uinf)\times S^2\right)^2$.
\begin{prop}The operator $H$ is symmetric with respect to the inner product $\langle \cdot, \cdot\rangle$ on the domain $C_0^{\infty}\left((-\infty,\uinf)\times S^2\right)^2$.\end{prop}
\begin{proof}Consider a solution $\Psi$ of $\eqref{cp2}$. Upon making the identification $\Psi = (\psi, i\psi_t)^T$, we know that $\psi$ solves $\eqref{cp}$. We have that $\langle \Psi,\Psi\rangle = E(\psi)$, and therefore that $\frac{d}{dt}\langle \Psi, \Psi\rangle = 0$ for solutions of $\eqref{cp2}$. On the other hand, 
\begin{align*}\frac{d}{dt}\langle \Psi, \Psi \rangle &= \langle \D_t \Psi, \Psi\rangle + \langle \Psi, \D_t\Psi\rangle \\
&=-i\langle H\Psi, \Psi\rangle + i\langle \Psi, H\Psi\rangle,\end{align*}
which shows that $\langle H\Psi, \Psi\rangle = \langle \Psi, H\Psi\rangle$ for any $\Psi$ solving $\eqref{cp2}$. Note that this expression holds for each $t$, and in particular, at $t = 0$. Thus, $\langle H\Psi_0, \Psi_0\rangle = \langle \Psi_0, H\Psi_0\rangle$. But the initial data $\Psi_0$ can be chosen arbitrarily in $C_0^{\infty}\left((-\infty,\uinf)\times S^2\right)^2$, which, after a simple polarization argument, shows that $H$ is symmetric on the space $C_0^{\infty}\left((-\infty,\uinf)\times S^2\right)^2$ with respect to the inner product $\langle \cdot, \cdot \rangle$.\end{proof}

We next observe that the only manifestation of the angular variables $(\theta,\phi)$ in the problem $\eqref{cp}$ is in the spherical Laplacian. Since any smooth function on $S^2$ can be expanded into an absolutely and uniformly convergent series in terms of spherical harmonics (c.f. \cite{courant1}), we may therefore write
\begin{equation}\Psi(t,u,\theta,\phi) = \sum_{l=0}^{\infty}\sum_{|m|\leq l} \Psi^{lm}(t,u)Y_{lm}(\theta,\phi),\label{spharmdecomp}\end{equation}
where the $Y_{lm}(\theta,\phi)$ are the spherical harmonics (i.e. $\Delta_{S^2} Y_{lm} = -l(l+1)Y_{lm}$) and this series converges uniformly and absolutely for each fixed $(t,u) \in R \times (-\infty,\uinf)$. Furthermore, we know that $\Psi^{lm} = (\psi_1^{lm}, \psi_2^{lm})^T$, where $\psi^{lm}_i = \langle \psi_i, Y_{lm}\rangle_{L^2(S^2)}$. It is clear therefore that $\Psi^{lm}(t,u)$ is smooth and for each $t$, $\Psi^{lm}(t,u) \in C_0^{\infty}\left(-\infty,\uinf\right)^2$.
Thus, for any $\Psi, \Gamma \in C_0^{\infty}\left((-\infty,\uinf)\times S^2\right)^2$, we can decompose the scalar product $\langle \Psi, \Gamma \rangle$ according to
\begin{equation*}\langle \Psi, \Gamma\rangle = \sum_{l=0}^{\infty}\sum_{|m|\leq l} \langle \Psi^{lm}, \Gamma^{lm}\rangle_l =\sum_{l=0}^{\infty}\sum_{|m|\leq l} \int_{-\infty}^{\uinf} r^4 \psi_2^{lm}\overline{\gamma_2^{lm}} + (\D_u\psi_1^{lm})\overline{(\D_u\gamma_1^{lm})} + \frac{r^2}{T^2}l(l+1)\psi_1^{lm}\overline{\gamma_1^{lm}}du,\end{equation*}
which follows from integrating by parts.

The action of the Hamiltonian also simplifies under the spherical harmonic decomposition: 
\begin{equation*}H\Psi(t,u,\theta,\phi) = \sum_{l=0}^{\infty}\sum_{|m|\leq l} H_l \Psi^{lm}(t,u)Y_{lm}(\theta,\phi),\end{equation*}
where
\begin{equation}H_l = \left(\begin{array}{cc} 0 & 1 \\ -\frac{1}{r^4}\D_u^2 + \frac{l(l+1)}{r^2 T^2} & 0\end{array}\right).\label{defofhamiltonian}\end{equation}
Therefore, the components $\Psi^{lm}$ in the spherical harmonic decomposition of $\Psi$ solve a reduced equation:
\begin{prop}Consider the solution $\Psi$ of $\eqref{cp2}$. The component functions $\Psi^{lm}$ in the spherical harmonic decomposition of $\Psi$ $\eqref{spharmdecomp}$ solve the reduced problem
\begin{equation}\begin{cases}i\D_t \Psi^{lm} = H_l \Psi^{lm} \text{ on } \R \times (-\infty,\uinf),\\ \Psi^{lm}(0,u) = \Psi_0^{lm} \in C_0^{\infty}(-\infty,\uinf)^2.\label{cp3}\end{cases}\end{equation}\end{prop}
\begin{proof}The proposition follows from the discussion above and the uniqueness of the spherical harmonic decomposition.\end{proof}
Our strategy therefore is to solve problem $\eqref{cp3}$ and then sum up according to $\eqref{spharmdecomp}$ to obtain a solution of $\eqref{cp2}$. We note as well that $H_l$ is symmetric with respect to the inner product $\langle \cdot, \cdot \rangle_l$ on the domain $C_0^{\infty}(-\infty,\uinf)^2$, since for $\Psi^{lm},\Gamma^{lm} \in C_0^{\infty}(-\infty,\uinf)^2$ we have
\begin{align*}\langle H_l \Psi^{lm}, \Gamma^{lm}\rangle_l &= \langle H(\Psi^{lm}Y_{lm}), \Gamma^{lm}Y_{lm}\rangle \\ &= \langle \Psi^{lm}Y_{lm}, H(\Gamma^{lm}Y_{lm})\rangle \\ &= \langle \Psi^{lm}, H_l \Gamma^{lm}\rangle_l.\end{align*}
This also implies that the energy $E_l(\Psi^{lm}) := \langle \Psi^{lm},\Psi^{lm}\rangle_l$ is conserved for smooth, compactly supported solutions of $\eqref{cp3}$, since we have
\begin{align*}\frac{d}{dt}\langle \Psi^{lm},\Psi^{lm}\rangle_l &= \langle \D_t \Psi^{lm},\Psi^{lm}\rangle_l + \langle \Psi^{lm}, \D_t \Psi^{lm}\rangle_l \\ &= -i\langle H_l\Psi^{lm},\Psi^{lm}\rangle_l + i\langle \Psi^{lm}, H\Psi^{lm}\rangle_l \\ &= 0,\end{align*}
by the symmetry of $H_l$.
\section{The Hamiltonian}Let us rewrite $H_l$ as
\begin{equation}H_l = \left(\begin{array}{cc} 0 & 1 \\ -\frac{1}{r^4}\D_u^2 + V_l(u) & 0\end{array}\right),\label{defofhamiltonian2}\end{equation}
where $V_l(u) = \frac{l(l+1)}{r^2 T^2}$. (Recall the arguments are $T = T(r)$ and $r = r(u)$.) We wish to construct a self-adjoint extension of $H_l$, and we therefore need to find a Hilbert space on which $H_l$ is densely defined. To that end, let us define $\mathscr{H}_{V_l,0}^1$ as the completion of $C_0^{\infty}(-\infty,\uinf)$ within the Hilbert space $\mathscr{H}^1_{V_l}(-\infty,\uinf) := \left\{\psi : \psi_u \in L^2(-\infty,\uinf) \text{ and } r^2V_l^{\frac{1}{2}}\psi \in L^2(-\infty, \uinf)\right\}$. Let us also define $\mathscr{H}_{r^2, 0}$ as the completion of $C_0^{\infty}(-\infty, \uinf)$ within the Hilbert space $\mathscr{H}_{r^2}(-\infty,\uinf):= \left\{\psi : r^2 \psi \in L^2(-\infty,\uinf)\right\}$. Finally, we define Hilbert space $\mathscr{H} := \mathscr{H}_{V_l,0}^1 \otimes \mathscr{H}_{r^2, 0}$ endowed with the inner product $\langle \cdot, \cdot \rangle_l$ to be the underlying Hilbert space for the Hamiltonian $H_l$.

We next construct a self-adjoint extension of $H_l$:
\begin{prop}The operator $H_l$ with domain $\mathscr{D}(H_l) = C_0^{\infty}(-\infty,\uinf)^2$ is essentially self-adjoint in the Hilbert space $\mathscr{H}$.\end{prop}
\begin{proof}To prove this, we use the following version of Stone's theorem (c.f. \cite{reedsimon1}, Sec. VIII.4):
\begin{thm}[{\bf Stone's Theorem}]Let $U(t)$ be a strongly continuous one-parameter unitary group on a Hilbert space $\mathscr{H}$. Then there is a self-adjoint operator $A$ on $\mathscr{H}$ so that $U(t) = e^{itA}$. Furthermore, let $\mathscr{D}$ be a dense domain which is invariant under $U(t)$ and on which $U(t)$ is strongly differentiable. Then $i^{-1}$ times the strong derivative of $U(t)$ is essentially self-adjoint on $\mathscr{D}$ and its closure is $A$.\label{st}\end{thm}

Now consider the Cauchy problem $\eqref{cp3}$. By the theory of symmetric hyperbolic systems, the problem $\eqref{cp3}$ has a unique, smooth, global solution $\Psi^{lm}$ that is compactly supported for all times $t$ (we prove this similarly to Theorem~\ref{existenceanduniqueness}). Thus, for $t\in \R$ we define the operators 
\[U(t): C_0^{\infty}(-\infty,\uinf)^2\mapsto C_0^{\infty}(-\infty,\uinf)^2 \text{ by }\]
\[U(t)\Psi^{lm}_0 = \Psi^{lm}(t) = \left(\psi^{lm}(t),i\D_t\psi^{lm}(t)\right)^T.\]
Note that $U(t)$ leaves the dense subspace $C_0^{\infty}(-\infty,\uinf)^2$ invariant for all times $t$ and also, by the energy conservation, the $U(t)$ are unitary with respect to the energy inner product and therefore extend to unitary operators on $\mathscr{H}$. The uniqueness of $\Psi^{lm}$ guarantees that $U(0) = I$ and $U(t)U(s) = U(t+s)$ for all $s,t \in \R$. Thus, the $U(t)$ form a one-parameter unitary group. The fact that the solutions are smooth in $t$ and $u$ guarantees that this group is strongly continuous on $\mathscr{H}$ and strongly differentiable on $C_0^{\infty}(-\infty,\uinf)^2$. Then, for $\gamma_1,\gamma_2 \in C_0^{\infty}(-\infty,\uinf)$, 
\[i^{-1}\lim_{h\to0} \frac{1}{h}\left(U(h)(\gamma_1,\gamma_2)^T - (\gamma_1,\gamma_2)^T\right) = -H_l (\gamma_1,\gamma_2)^T.\]
Thus, by Stone's theorem, $H_l$ is essentially self-adjoint on $\mathscr{H}$ with self-adjoint closure $\bar{H}_l$ and $U(t) = e^{-it\bar{H}_l}$.\end{proof}

To obtain a representation of the solution $\Psi^{lm}$ of $\eqref{cp3}$, we will use Stone's formula which relates the spectral projections of a self-adjoint operator to the resolvent. We recall Stone's formula in the following theorem:
\begin{thm}[{\bf Stone's Formula}]For a self-adjoint operator $A$, the following holds
\begin{equation}\frac{1}{2}\left(P_{[a,b]} + P_{(a,b)}\right) = \lim_{\eps \searrow 0} \frac{1}{2\pi i}\int_a^b \left[ (A - \lambda -i\eps)^{-1} - (A - \lambda + i\eps)^{-1}\right]d\lambda,\label{stonesformula}\end{equation}
where the limit is taken in the strong operator topology.\label{stoneformula}\end{thm}
We refer to \cite{reedsimon1}, chapter VII for a proof. According to Stone's formula, to understand the spectral projections of $\ham$, we must investigate the resolvent operator $(\ham - \omega)^{-1}: \mathscr{H} \mapsto \mathscr{H}$. Since $\ham$ is self-adjoint, it follows immediately that $(\ham - \omega)^{-1}$ exists for each $\omega \in \C\setminus \R$. So let us fix $\omega \in \C \setminus\R$ and consider the eigenvalue equation 
\begin{equation}\ham \Phi = \omega \Phi.\label{eigenvalueequation}\end{equation}
Note that since $\omega \not\in \sigma(\ham)$, this equation does not have solutions in $\mathscr{H}$. Nonetheless, we will be able to construct the resolvent out of special solutions of this equation. To that end, let us observe that $\eqref{eigenvalueequation}$ is equivalent to the differential equation 
\begin{equation}-\zeta''(u) - \omega^2r^4\zeta + \frac{r^2}{T^2}l(l+1)\zeta = 0\label{evalueequation}\end{equation}
on the interval $(-\infty,\uinf)$ where $\zeta = \phi_1 \text{ or } \phi_2$. This ODE is difficult to solve explicitly, so let us use the coordinate $s= s(u)$ defined in $\eqref{scoord}$ and define 
\begin{equation}\eta(s) = r(u(s))\zeta(u(s)).\label{etadef}\end{equation}
Inserting these into $\eqref{evalueequation}$, we obtain the equivalent ODE 
\begin{equation}-\eta''(s) - \omega^2\eta(s) + \left(\frac{l(l+1)}{r^2T^2} - \frac{1}{rT^2K^2}\left(\frac{T'}{T} + \frac{K'}{K}\right)\right)\eta(s)=0.\label{evalueequation2}\end{equation}
To investigate this ODE, we need to look at the potential 
\begin{equation}W_l(s):= \left(\frac{l(l+1)}{r^2T^2} - \frac{1}{rT^2K^2}\left(\frac{T'}{T} + \frac{K'}{K}\right)\right).\label{potential}\end{equation}
Invoking the asymptotics $\eqref{farfieldasymptotics1}$ and $\eqref{farfieldasymptotics2}$, we see that as $s$ tends to infinity, $|W_l(s)| = O\left(\frac{l(l+1)}{s^2}\right)$ for $l\neq 0$ and $|W_0(s)| = O\left(\frac{1}{s^3}\right)$. We note that for large $r$, we have $s = r + O(\log r)$, and for small $r$ we have $s = c_1c_2log(r - r_0) + O(1)$, using the asymptotics $\eqref{eventhorizonasymptotics1}$, $\eqref{eventhorizonasymptotics2}$ as well as $\eqref{s(r)}$. Noting also that $W_l(s) = O(r - r_0)$ near the horizon, the previous comments yield that $|W_l(s)| \leq \alpha_1 e^{\alpha_2 s}$ as $s \to -\infty$ for some constants $\alpha_1,\alpha_2>0$. 

We now return to equation $\eqref{evalueequation2}$ and prescribe asymptotic boundary conditions to determine a pair of fundamental solutions. In the case $\Im(\omega) >0$, we require
 \begin{equation}\lim_{s\to -\infty}e^{i\omega s}\eta_{1,\omega}(s) = 1, \text{ and } \lim_{s\to-\infty}\left(e^{i\omega s} \eta_{1,\omega}(s)\right)' = 0\label{leftposbcs}\end{equation} 
and
\begin{equation}\lim_{s\to \infty}e^{-i\omega s}\eta_{2,\omega}(s) = 1, \text{ and } \lim_{s\to\infty}\left(e^{-i\omega s} \eta_{2,\omega}(s)\right)' = 0,\label{rightposbcs}\end{equation}
whereas in the case $\Im(\omega)<0$, we require
 \begin{equation}\lim_{s\to -\infty}e^{-i\omega s}\eta_{1,\omega}(s) = 1, \text{ and } \lim_{s\to-\infty}\left(e^{-i\omega s} \eta_{1,\omega}(s)\right)' = 0\label{leftnegbcs}\end{equation} 
and
\begin{equation}\lim_{s\to \infty}e^{i\omega s}\eta_{2,\omega}(s) = 1, \text{ and } \lim_{s\to\infty}\left(e^{i\omega s} \eta_{2,\omega}(s)\right)' = 0.\label{rightnegbcs}\end{equation}
Now these two solutions $\eta_{1,\omega}, \eta_{2,\omega}$ must be linearly independent, for if they were not, there would be a nonzero vector in the kernel of $(\ham - \omega)^{-1}$. But since $\ham$ is essentially self-adjoint, the spectrum is contained on the real line, and thus, for $\omega \in \C\setminus\R$, the kernel of $(\ham - \omega)^{-1}$ is trivial. Thus, $\etaa$ and $\etab$ form a fundamental set of solutions of $\eqref{evalueequation2}$, and the Wronskian $w(\etaa,\etab) := \etaa(s)\etab'(s) - \etaa'(s)\etab(s)$ is non-vanishing. We note also that an easy calculation shows that $w(\etaa,\etab)$ is independent of $s$.

Let us now construct the solutions $\etaa,\etab$ (for notational purposes, in this section we will write $\etaa = \eta^1(\lambda,\omega,s), \etab = \eta^2(\lambda,\omega,s)$, where $\lambda = l+\frac{1}{2}$; the $\lambda$ dependence in what follows can be important in a more general setting, so to make this as general as possible, we make explicit the $\lambda$ dependence). We cite \cite{potentialscattering} for the basic idea of this construction. We focus first on the solution with boundary conditions at $s = \infty$, and we restrict ourselves for the moment to $\Im \omega \leq 0$, $w\neq 0$.
We first write the ODE $\eqref{evalueequation2}$ as 
\begin{equation*}\eta''(s) + \left(\omega^2 - \frac{\lambda^2 - \frac{1}{4}}{s^2}\right)\eta(s) = \left(\left(\lambda^2 - \frac{1}{4}\right)\left[\frac{1}{r^2T^2} - \frac{1}{s^2} \right] - \frac{1}{rT^2K^2}\left(\frac{T'}{T} + \frac{K'}{K}\right)\right)\eta(s),\end{equation*}
and we find the Green's function for the operator on the left-hand side with zero boundary conditions at $ s=\infty$ is given by
\begin{equation}B(\lambda,\omega,s,y) = H(y-s)\frac{i}{2\omega}\left(\eta^2_0(\lambda,\omega,y)\eta^2_0(\lambda,-\omega,s) - \eta^2_0(\lambda,\omega,s)\eta^2_0(\lambda,-\omega,y)\right),\label{defofB}\end{equation}
where
\begin{equation}\eta^2_0(\lambda,\omega,s) = \left(\frac{1}{2}\pi \omega s\right)^{\frac{1}{2}}e^{-\frac{i\pi}{2}\left(\lambda + \frac{1}{2}\right)}H_{\lambda}^{(2)}(\omega s),\label{eta20def}\end{equation}
$H$ is the usual Heaviside function, and $H_{\lambda}^{(2)}$ is the Hankel function of the second kind (we reference \cite{specialfunctions} and  \cite{watsonbesselfunctions} for information about the Hankel functions). Note that $\lim_{s\to \infty}\eta^2_0(\lambda,\omega,s)e^{i\omega s} = 1$. Thus, if we require 
\begin{equation}\lim_{s \to \infty}\eta^2(\lambda,\omega,s)e^{i\omega s} = 1,\label{eta2bc}\end{equation}
then the equivalent integral equation for $\eta^2$ is
\begin{equation}\eta^2(\lambda,\omega,s) = \eta^2_0(\lambda,\omega,s) + \int_s^{\infty}B(\lambda,\omega,s,y)W(y)\eta^2(\lambda,\omega,y)dy.\label{eta2intequation}\end{equation}

We wish to solve this as a perturbation series, so we write 
\begin{equation}\eta^2(\lambda,\omega,s) = \sum_{n=0}^{\infty}\eta^2_{n}(\lambda,\omega,s),\label{eta2series}\end{equation}
where
\begin{equation}\eta^2_{n+1} = \int_s^{\infty}B(\lambda,\omega,s,y)W(y)\eta^2_n(\lambda,\omega,y)dy.\label{eta2ndef}\end{equation}
To address convergence, we note that it is shown in appendix A of \cite{potentialscattering} that for $0<s<y$ we have
\begin{equation}|\eta^2_0(\lambda,\omega,s)| \leq C \left(\frac{|\omega|s}{1 + |\omega|s}\right)^{-\lambda + \frac{1}{2}}e^{(\Im \omega)s}\label{eta20bound}\end{equation}
and
\begin{equation}|B(\lambda,\omega,s,y)| \leq C e^{|\Im \omega| y + (\Im \omega)s}\left(\frac{y}{1 + |\omega| y}\right)^{\lambda + \frac{1}{2}}\left(\frac{s}{1 + |\omega|s}\right)^{-\lambda + \frac{1}{2}}\label{Bbound}\end{equation}
where $C$ depends on $\lambda$. It is easy to show then by induction that
\begin{equation}|\eta_n^2(\lambda,\omega,s)| \leq C \frac{(CQ(s))^n}{n!}\left(\frac{|\omega|s}{1 + |\omega|s}\right)^{-\lambda + \frac{1}{2}}e^{(\Im \omega)s},\label{eta2nbound}\end{equation}
where
\begin{equation}Q(s) = \int_s^{\infty}\frac{y|W(y)|}{1 + |\omega| y}e^{(|\Im \omega| + \Im \omega )y}dy.\label{defofQ}\end{equation}
Note that for $\Im \omega \leq 0$, $Q$ is finite for all $s \in [0,\infty)$ and indeed $\|Q\|_{L^1([0,\infty))} < \infty$, owing to the integrability of $W$ and our requirement that $\Im \omega \leq 0$. Thus $\eta^2$ exists (for $\Im \omega \leq 0$ and $\omega \neq 0$), and the following bounds are obvious
\begin{equation}|\eta^2(\lambda,\omega,s)| \leq Ce^{(\Im \omega)s}\left(\frac{|\omega|s}{1 + |\omega|s}\right)^{-\lambda + \frac{1}{2}}e^{CQ(s)},\label{eta2bound}\end{equation}
and
\begin{equation}|\eta^2(\lambda,\omega,s) - \eta^2_0(\lambda,\omega,s)| \leq Ce^{(\Im \omega)s}\left(\frac{|\omega|s}{1 + |\omega|s}\right)^{-\lambda + \frac{1}{2}}(e^{CQ(s)}-1).\label{eta2minuseta20bound}\end{equation}
It is straightforward to show that $\eta^2$ is smooth in $s$ for fixed $\omega$, analytic in $\omega$ for fixed $s$ (for $\Im \omega < 0$), unique, and that $\eta^2$ solves the ODE $\eqref{evalueequation2}$. Furthermore, we easily obtain the following estimates:
\begin{equation}\left|\frac{d}{ds}\eta^2_0(\lambda,\omega,s)\right| \leq C|\omega| e^{(\Im \omega)s}\left(\frac{|\omega|s}{1 + |\omega|s}\right)^{-\lambda - \frac{1}{2}}\label{deta20bound}\end{equation}
and
\begin{equation}\left|\frac{d}{ds}\eta^2(\lambda,\omega,s) - \frac{d}{ds}\eta^2_0(\lambda,\omega,s)\right| \leq C\left(\frac{|\omega|s}{1 + |\omega|s}\right)^{-\lambda - \frac{1}{2}}e^{(\Im \omega)s}\int_s^{\infty}\left(\frac{|\omega|y}{1 + |\omega|y}\right)^{-\lambda + \frac{1}{2}}e^{CQ(y)}|W(y)|dy.\label{deta2minusdeta20bound}\end{equation}

From $\eqref{eta2bound}$ we see a possible singularity in $\eta^2$ at $\omega = 0$, but this singularity is removable. Indeed, repeating the above construction with the initial function $\eta^{2,0}_0(\lambda,\omega,s) = \omega^{\lambda - \frac{1}{2}}\left(\frac{1}{2}\pi \omega s\right)^{\frac{1}{2}}e^{-\frac{i\pi}{2}\left(\lambda + \frac{1}{2}\right)}H_{\lambda}^{(2)}(\omega s)$ yields a solution $\eta^{2,0}$ of the integral equation
\begin{equation*}\eta^{2,0}(\lambda,\omega,s) = \eta^{2,0}_0(\lambda,\omega,s) + \int_s^{\infty}B(\lambda,\omega,s,y)W(y)\eta^{2,0}(\lambda,\omega,y)dy.\label{eta2intequation}\end{equation*}
This solution satisfies the boundary conditions $\lim_{s\to \infty}\eta^{2,0}(\lambda,\omega,s)e^{i\omega s} = \omega^{\lambda - \frac{1}{2}}$ and it is continuous in the region $\Im \omega \leq 0$. Finally, $\eta^{2,0}$ can also be obtained from $\eta^2$ in the sense that $\omega^{\lambda - \frac{1}{2}} \eta^2 = \eta^{2,0}$ (by uniqueness).

So we have solved the ODE $\eqref{evalueequation2}$ with boundary conditions at $s = \infty$ for $\Im \omega \leq 0$. For $\Im \omega >0$, we obtain a solution of $\eta^2(\lambda,\omega,s)$ of this BVP by defining $\eta^2(\lambda,\omega,s) = \overline{\eta^2(\lambda, \bar{\omega},s)}$. The uniqueness guarantees that this is indeed a solution and it is easy to check that this function is well-behaved as $\omega \to 0$.

A similar construction produces a solution $\eta^1(\lambda,\omega,s)$ of $\eqref{evalueequation2}$ with boundary conditions at $s = -\infty$ which is also smooth in $s$, analytic in $\omega$, and unique. We can, furthermore, extend this solution to $\omega  = 0$ as above.

We can now use these solutions to construct the resolvent. For ease of notation, we will let $\etaa(s) = \eta^1(\lambda,\omega,s)$ and $\etab(s) = \eta^2(\lambda,\omega,s)$, since we are considering $l$ and therefore $\lambda$ to be fixed. Then, we use the the definition $\eqref{etadef}$ to obtain from $\etaa,\etab$ two solutions $\zetaa,\zetab$ of $\eqref{evalueequation}$, and in the case $\omega=0$, we again use $\eqref{etadef}$ to obtain solutions $\zeta_{1,0},\zeta_{2,0}$. It's easy to see that $w(\etaa,\etab) = w(\zetaa,\zetab)$, and therefore, it follows that for $\Im \omega \neq 0$, $\zetaa,\zetab$ form a pair of fundamental solutions for the ODE $\eqref{evalueequation}$ with non-vanishing Wronskian. Thus we may define the following function
\begin{equation}h_{\omega}(u,v):= -\frac{1}{w(\zetaa,\zetab)}\begin{cases}\zetaa(u)\zetab(v), & u\leq v \\ \zetaa(v)\zetab(u), & u > v.\end{cases}\label{defofh}\end{equation}
An easy calculation shows that $h_{\omega}(u,v)$ satisfies the distributional equations
\begin{equation}\left(-d_u^2 - r^4\omega^2 + \frac{r^2}{T^2}l(l+1)\right)h_{\omega}(u,v) = \delta(u-v) = \left(-d_v^2 - r^4\omega^2 + \frac{r^2}{T^2}l(l+1)\right)h_{\omega}(u,v)\label{disteqn}\end{equation}
where the arguments on the left are $r = r(v)$ and $r = r(u)$ on the right. We next use the function $h_{\omega}(u,v)$ to construct the resolvent $(\ham - \omega)^{-1}$. We note here that this argument is similar to an argument in \cite{thecauchyproblemforthewaveequationintheschwarzschildgeometry}.
\begin{prop}For any $\omega \in \C\setminus \R$, the resolvent $(\ham - \omega)^{-1}$ can be represented as an integral operator with kernel 
\begin{equation}k_{\omega}(u,v) = \delta(u-v)\left(\begin{array}{cc} 0 & 0 \\ 1 & 0\end{array}\right) + r^4(v)h_{\omega}(u,v)\left(\begin{array}{cc} \omega & 1\\ \omega^2 & \omega\end{array}\right).\label{resolventkernel}\end{equation}\end{prop}
\begin{proof}Consider the integral operator $K_{\omega}$ with kernel given by $k_{\omega}(u,v)$ acting on the domain $\mathscr{D}(K_{\omega}):= \left\{ (\ham - \omega)\Psi : \Psi \in C_0^{\infty}(-\infty,\uinf)^2\right\}$. We claim first that $\mathscr{D}(K_{\omega})$ is a dense subset of $\mathscr{H}$. To prove this, let $\xi \in \mathscr{H}$ be arbitrary. Because the resolvent exists, $(\ham - \omega): \mathscr{D}(\ham) \mapsto \mathscr{H}$ is onto. Thus, there exists $\gamma \in \mathscr{D}(\ham)$ so that $(\ham - \omega)\gamma = \xi$. Since $\ham$ is the closure of $H_l$, there is a sequence $\{\gamma_n\} \subset C_0^{\infty}(-\infty,\uinf)$ so that $\gamma_n \to \gamma$ and $\ham \gamma_n \to \ham \gamma$ as $n\to \infty$. Thus, $\{(\ham - \omega)\gamma_n\}$ converges to $(\ham - \omega)\gamma = \xi$, and thus, $\mathscr{D}(K_{\omega})$ is dense in $\mathscr{H}$. 

Now, for an arbitrary $\Gamma = (\gamma_1,\gamma_2)^T \in C_0^{\infty}(-\infty,\uinf)^2$, we have 
\begin{align*}(S_{\omega}(\ham - &\omega) \Gamma)(u) := \int_{-\infty}^{\uinf}k_{\omega}(u,v)(\ham - \omega)\Gamma(v)dv\\
&=\int_{-\infty}^{\uinf}\left[ \delta(u-v)\left(\begin{array}{cc} 0 & 0 \\ 1 & 0\end{array}\right) + r^4(v)h_{\omega}(u,v)\left(\begin{array}{cc} \omega & 1\\ \omega^2 & \omega\end{array}\right)\right]\left(\begin{array}{c}-\omega \gamma_1 + \gamma_2\\ \left(\frac{1}{r^4}d_u^2 + \frac{l(l+1)}{r^2T^2}\right)\gamma_1 - \omega\gamma_2\end{array}\right)dv \\
&=(0, -\omega\gamma_1 + \gamma_2)^T + \int_{-\infty}^{\uinf}h_{\omega}(u,v)\left(\begin{array}{c}\left(-d_u^2 - r^4\omega^2 + \frac{r^2}{T^2}l(l+1)\right)\gamma_1\\ \left(-d_u^2 - r^4\omega^2 + \frac{r^2}{T^2}l(l+1)\right)\gamma_1\omega\end{array}\right)dv\\
&=(\gamma_1,\gamma_2)^T\\
&= \Gamma,\end{align*}
where we have used $\eqref{disteqn}$. Thus, $K_{\omega}(\ham - \omega) = I$ on $C_0^{\infty}(-\infty,\uinf)^2$ and hence $K_{\omega} = (\ham - \omega)^{-1}$ on $\mathscr{D}(K_{\omega})$. Since $(\ham - \omega)^{-1}$ is bounded and $\mathscr{D}(K_{\omega})$ is dense, the claim follows.\end{proof}

We can now apply Stone's formula to $\ham$ to get, for each $\Psi \in \mathscr{H}$ that
\begin{align}\frac{1}{2}\left(P_{[a,b]} + P_{(a,b)}\right)\Psi(u) &= \lim_{\eps\searrow 0} \frac{1}{2\pi i} \int_a^b \left[(\ham - (\omega + i\eps))^{-1} - (\ham - (\omega -i\eps))^{-1}\right]\Psi(u)d\omega \notag\\
&=\lim_{\eps\searrow 0} \frac{1}{2\pi i} \int_a^b \left(\int_{-\infty}^{\uinf}\left(k_{\omega + i\eps}(u,v) - k_{\omega - i\eps}(u,v)\right)\Psi(v)dv\right)d\omega,\label{stoneformula2}\end{align}
where the limit is taken with respect to the norm in $\mathscr{H}$.
\section{A Representation Formula}
In this section we will obtain an integral representation formula for the solution of the Cauchy problem $\eqref{cp3}$ via $\eqref{stoneformula2}$. We begin first with a proposition:
\begin{prop}The Wronskian $w(\zetaa,\zetab)$ does not vanish for $\omega \in \R\setminus\{0\}$. This remains valid in the case $\omega = 0$ for the solutions $\zeta_{1,0}, \zeta_{2,0}$.\end{prop}
\begin{proof}We first note, again, that $w(\etaa,\etab) = w(\zetaa,\zetab)$ and it therefore suffices to the lemma for the $\eta$ solutions or the $\zeta$ solutions. For the $\omega = 0$ case, we observe that $\zeta_{1,0},\zeta_{2,0}$ solve the ODE 
\begin{equation}\zeta''(u) = \frac{r^2}{T^2}l(l+1)\zeta(u),\label{omega0equation}\end{equation}
subject to the asymptotic boundary conditions
\begin{align*}&\lim_{u\nearrow \uinf} s(u)^l \zeta_0 r(u) = (-i)^l(2l-1)!!\\
&\lim_{u\to -\infty} \zeta_{1,0}(u)r(u) = 1.\end{align*}
Thus, equation $\eqref{omega0equation}$ with the asymptotic boundary conditions implies that the solution $\zeta_{1,0}$ is convex. Similarly, since the solution $\zeta_{2,0}$ must be either real or purely imaginary depending on whether $l$ is odd or even, $\eqref{omega0equation}$ implies that either $\Re(\zeta_{2,0})$ or $\Im(\zeta_{2,0})$ is strictly convex or concave (again depending on $l$). In any case, this observation coupled with the asymptotic boundary conditions imply that $\zeta_{1,0}$ and $\zeta_{2,0}$ are linearly independent and thus that $w(\zeta_{1,0},\zeta_{2,0})\neq 0$.

In the case $\omega \in \R\setminus\{0\}$, it's easy to show, using the asymptotic boundary conditions $\eqref{leftposbcs} - \eqref{rightnegbcs}$, that $w(\Re(\eta_{j,\omega}),\Im(\eta_{j,\omega})) \neq 0$ for $j=1,2$. Next, for $j\in \{1,2\}$, consider $y_j:= \frac{\eta_{j,\omega}'}{\eta_{j\omega}}$. An easy calculation shows that 
\[\Im(y_j) = \frac{w(\Re(\eta_{j,\omega}),\Im(\eta_{j,\omega}))}{|\eta_{j,\omega}|^2}.\] Note that $y_j$ is well-defined since $w(\Re(\eta_{j,\omega}),\Im(\eta_{j,\omega})) \neq 0$. Thus, $\Im(y_j) \neq 0$ and either $\Im(y_j)>0$ or $\Im(y_j) <0$ by continuity for all $s\in (-\infty,\infty)$. Moreover, using the boundary conditions again, it's easy to show that $\Im(y_1)$ and $\Im(y_2)$ have different signs. Therefore
\begin{equation*}w(\etaa,\etab) = \etaa\etab(y_2 - y_1) \neq 0,\end{equation*}
and hence, $w(\zetaa,\zetab)\neq 0$.\end{proof}
As a consequence, we have
\begin{cor}The function $h_{\omega}(u,v)$ defined in $\eqref{defofh}$ is continuous in $(\omega,u,v)$ for $\omega \in \{\Im(\omega)\leq0\}$ and $(u,v)\in (-\infty,\uinf)^2$.\end{cor}
\begin{proof}Note that the analytic dependence on $\omega$ of the ODE $\eqref{evalueequation}$ guarantees that the $\zeta$ solutions depend at least continuously on $\omega$ for $\Im(\omega)\leq 0$. Moreover, since $h_{\omega}(u,v)$ is invariant under the substitution $\omega^l\zetab$ for $\zetab$, the previous proposition yields the claim.\end{proof}
Next, observe that the definitions of $\eta^1,\eta^2$ for $\Im \omega>0$ imply that $\overline{h_{\omega}(u,v)} = h_{\bar{\omega}}(u,v)$, and hence, $\overline{k_{\omega}(u,v)} = k_{\bar{\omega}}(u,v)$. We can then simplify $\eqref{stoneformula2}$ to read
\begin{equation}\frac{1}{2}\left(P_{[a,b]} + P_{(a,b)}\right)\Psi(u) = \lim_{\eps\searrow 0} -\frac{1}{\pi}\int_a^b\left(\int_{-\infty}^{\uinf}\Im(k_{\omega - i\eps}(u,v))\Psi(v)dv\right)d\omega,\label{stoneformula3}\end{equation}
where this converges in $\mathscr{H}$-norm. Since the integrand is continuous and $\Psi \in C_0^{\infty}(-\infty,\uinf)^2$, for any bounded interval $[a,b]$ we are integrating a continuous integrand over a compact region, and if we consider this limit as a pointwise limit in $u$, then for any fixed $u$ we may exchange the limit and the integration (by Lebesgue's Dominated Convergence Theorem). This observation coupled with the norm convergence yields
\begin{equation}\frac{1}{2}\left(P_{[a,b]} + P_{(a,b)}\right)\Psi(u) = -\frac{1}{\pi}\int_a^b\left(\int_{\text{supp} \Psi} \Im(k_{\omega}(u,v))\Psi(v)dv\right)d\omega.\label{stoneformula4}\end{equation}
Note that this yields that $P_{\{a\}} = 0$ for any $a\in \R$, and thus that $P_{[a,b]} = P_{(a,b)}$. This in turn implies that the spectrum $\sigma(\ham)$ is absolutely continuous. In particular, this yields
\begin{equation}P_{(a,b)}\Psi(u) = -\frac{1}{\pi}\int_a^b\left(\int_{\text{supp} \Psi} \Im(k_{\omega}(u,v))\Psi(v)dv\right)d\omega\label{stoneformula5}\end{equation}
for any $\Psi\in C_0^{\infty}(-\infty,\uinf)^2$ and any bounded interval $(a,b)$.

We would next like to rewrite the integrand in $\eqref{stoneformula5}$ in a more useful form. To this end, let us observe that for $\omega \in \R\setminus\{0\}$, the pair $\{\zetaa, \overline{\zetaa}\}$ forms a fundamental system for the ODE $\eqref{evalueequation}$. Therefore, there exist constants (constant in $u,v$) $\lambda(\omega),\mu(\omega)$ so that 
\begin{equation}\zetab(u) = \lambda(\omega)\zetaa(u) + \mu(\omega)\overline{\zetaa}(u)\label{transcoefficients}\end{equation}
for $\omega \in \R\setminus\{0\}$. From the boundary conditions $\eqref{leftposbcs} - \eqref{rightnegbcs}$ (and the fact that $w(\etaa,\etab) = w(\zetaa,\zetab)$), it's easy to see that $w(\zetaa,\zetab) = -2i\omega\mu(\omega)$. Now, let us make the following definitions
\begin{equation}\gammaa(u) = \Re(\zetaa(u)), \text{ and } \gammab(u) = \Im(\zetaa(u)),\label{gammadef}\end{equation} 
as well as
\begin{equation}\Gammaa(u) = (\gamma_{a,\omega}(u),\omega\gamma_{a,\omega}(u))^T.\label{Gammadef}\end{equation}
Then for $\omega \neq 0$, an easy calculation shows that 
\begin{equation}\Im(h_{\omega}(u,v)) = -\frac{1}{2\omega}\sum_{a,b =1}^2\alpha_{ab}(\omega)\gamma_{a,\omega}(u)\gamma_{b,\omega}(v),\label{simpleh}\end{equation}
where the coefficients are given by
\begin{equation}\alpha_{11}(w) = 1 + \Re\left(\frac{\lambda(\omega)}{\mu(\omega)}\right), \alpha_{22}(\omega) = 1 - \Re\left(\frac{\lambda(\omega)}{\mu(\omega)}\right), \text{ and } \alpha_{12}(\omega) = \alpha_{21}(\omega) = -\Im\left(\frac{\lambda(\omega)}{\mu(\omega)}\right).\label{alphadef}\end{equation}
By continuity, this expression extends to $\omega = 0$, and we can therefore write
\begin{align}\int_{\text{supp}\Psi}\Im(k_{\omega}(u,v))\Psi(v)dv &= -\frac{1}{2\omega}\int_{\text{supp}\Psi} r^4 \sum_{a,b=1}^2 \alpha_{ab}(\omega)\gamma_{a,\omega}(u)\gamma_{b,\omega}(v)\left(\begin{array}{cc} \omega & 1 \\ \omega^2 & \omega \end{array}\right)\Psi(v)dv \notag\\
&= -\frac{1}{2\omega^2}\sum_{a,b=1}^2 \alpha_{ab}(\omega)\Gammaa(u)\int_{\text{supp}\Psi}(\omega^2\gammab(v)\psi_1(v) + \omega\gammab(v)\psi_2(v))r^4dv \notag\\
&=-\frac{1}{2\omega^2}\sum_{a,b=1}^2\alpha_{ab}(\omega)\Gammaa(u)\langle \Gammab,\Psi\rangle_l,\label{rewriteintegrand}\end{align}
where we have used the fact that $\omega^2 r^4\gammab(v) = (-d_v^2 + r^4V_l(v))\gammab(v)$ and we integrated by parts. We also note that the inner product above is well-defined because $\Psi \in C_0^{\infty}(-\infty,\uinf)^2$.

We use the above argument to finally obtain a representation formula for the solution $\Psi^{lm}$ of the Cauchy problem $\eqref{cp3}$.
\begin{prop}The solution $\Psi^{lm}$ of the Cauchy problem $\eqref{cp3}$ can be represented as 
\begin{align}\Psi^{lm}(t,u) &= e^{-it\ham}\Psi_0^{lm}(u)\notag\\
&= \frac{1}{2\pi}\int_{\R}e^{-i\omega t} \frac{1}{\omega^2}\sum_{a,b=1}^2\alpha_{ab}(\omega)\Gammaa(u)\langle \Gammab, \Psi_0\rangle_ld\omega,\label{repformula}\end{align}
where the integral converges in norm in $\mathscr{H}$.\end{prop}
\begin{proof}Using $\eqref{rewriteintegrand}$ in $\eqref{stoneformula5}$ and applying the spectral theorem, for any $n\in \N$ we have
\begin{equation*}e^{-it\ham}P_{(-n,n)}\Psi_0(u) = \frac{1}{2\pi}\int_{-n}^ne^{-i\omega t}\frac{1}{\omega^2}\sum_{a,b=1}^2\alpha_{ab}(\omega)\Gammaa(u)\langle \Gammab,\Psi_0^{lm}\rangle_ld\omega.\end{equation*}
Furthermore, since $e^{-it\ham}$ is unitary, we have 
\begin{equation*}e^{-it\ham}P_{(-n,n)}\Psi_0 \to e^{-it\ham}\Psi_0\end{equation*}
in the $\mathscr{H}$-norm as $n\to \infty$.\end{proof}
\section{Decay}
We now obtain decay results from the representation formula $\eqref{repformula}$. To this end, we state a proposition.
\begin{prop}For fixed $u\in (-\infty,\uinf)$, the integrand in the representation $\eqref{repformula}$ is in $L^1(\R,\C^2)$ as a function of $\omega$. In particular, the representation $\eqref{repformula}$ holds pointwise for each $u\in (-\infty,\uinf)$.\end{prop}
\begin{proof}Since the integrand is continuous, we need only to analyze it for $|\omega|\gg 1$. We must therefore investigate the asymptotic behavior of $\zetaa(u)$ for $|\omega|\gg 1$, but according to the definition $\eqref{etadef}$, we shall first analyze the asymptotic behavior of $\etaa$. To do this, we refer to the proof of existence of $\etaa$ in Theorem 5.1 of \cite{thecauchyproblemforthewaveequationintheschwarzschildgeometry} where $\etaa$ is constructed via a perturbation series
\begin{equation}\etaa(s) = \sum_{k=0}^{\infty}\etaa^{(k)}(s),\label{etaseries}\end{equation}
where $\etaa^{(0)}(s) = e^{i\omega s}$ and 
\[\etaa^{(k+1)}(s) = -\int_{-\infty}^s\frac{1}{\omega}\sin(\omega(s-\tilde{s}))W_l(s)\etaa^{(k)}(\tilde{s})d\tilde{s}.\]
We then have the following estimate
\[\left|\etaa^{(k+1)}(s)\right| \leq \int_{-\infty}^s \frac{1}{|\omega|}|W_l(\tilde{s})|\cdot|\etaa^{(k)}(\tilde{s})|d\tilde{s},\]
and if we assume, by way of induction, that 
\begin{equation}\left|\etaa^{(k)}(s)\right|\leq \frac{1}{k!}\left(\int_{-\infty}^s \frac{1}{|\omega|}|W_l(\tilde{s})|d\tilde{s}\right)^k,\label{ih}\end{equation}
we then have 
\begin{align*}\left|\etaa^{(k+1)}(s)\right| &\leq \int_{-\infty}^s \frac{1}{|\omega|}|W_l(\tilde{s})|\frac{1}{k!}\left(\int_{-\infty}^{\tilde{s}} \frac{1}{|\omega|}|W_l(\hat{s})|d\hat{s}\right)^k\\
&=\int_{-\infty}^s\frac{d}{d\tilde{s}}\left[\frac{1}{(k+1)!}\int_{-\infty}^{\tilde{s}}\frac{1}{|\omega|}|W_l(\hat{s})|d\hat{s}\right]^{k+1}d\tilde{s}\\
&=\frac{1}{(k+1)!}\left(\int_{-\infty}^s \frac{1}{|\omega|}|W_l(\tilde{s})|d\tilde{s}\right)^{k+1}.\end{align*}
Since the induction hypothesis $\eqref{ih}$ also holds for $k=0$, $\eqref{ih}$ holds for each $k \in \N$.

Therefore, we have the following estimate on $\etaa$ from $\eqref{etaseries}$:
\begin{align*}|\etaa(s)| &\leq \sum_{k=0}^{\infty}\frac{1}{k!}\left(\int_{-\infty}^s \frac{1}{|\omega|}|W_l(\tilde{s})|d\tilde{s}\right)^k \\
&\leq e^{\frac{1}{|\omega|}\|W_l\|_{L^1}} \\
&\leq 1 + O\left(\frac{1}{|\omega|}\right),\end{align*}
for $|\omega|\gg 1$, since we know $\|W_l\|_{L^1} < \infty$. Next we analyze $\langle \Psi_0^{lm},\Gammab\rangle_l$. We have
\[\langle \Psi_0^{lm}, \Gammab\rangle_l = \int_{\text{supp}\Psi_0^{lm}}r^4(\psi_0^{lm})_2\omega\gammab(u) + (\psi_0^{lm})_1'\gammab'(u) + \frac{r^2}{T^2}l(l+1)(\psi_0^{lm})_1\gammab du,\]
but since $\gammab$ solves the ODE $\eqref{evalueequation}$, we have 
\[\gammab = \frac{1}{\omega^2r^4}\left(-\gammab'' + \frac{r^2}{T^2}l(l+1)\gammab\right).\]
Substituting this in the expression above and integrating by parts twice yields
\begin{align*}\langle \Psi_0^{lm}, \Gammab\rangle_l &= \frac{1}{\omega^2}\int_{\text{supp}\Psi_0^{lm}}-\gammab(u)\left(\omega(\psi_0^{lm})_2  - \frac{1}{r^4}(\psi_0^{lm})_1'' + \frac{l(l+1)}{r^2T^2}(\psi_0^{lm})_1\right)'' \\
&\hspace{.3in}+ \gammab(u)\frac{r^2}{T^2}l(l+1)\left(\omega (\psi_0^{lm})_2 - \frac{(\psi_0^{lm})_1''}{r^4} + \frac{l(l+1)}{r^2T^2}\right)du\end{align*}
Since $r,T$ are smooth and $\Psi_0^{lm}\in C_0^{\infty}(-\infty,\uinf)^2$, we can iterate this argument as many times as we like to obtain arbitrary polynomial decay in $\omega$.

It remains for us to analyze the coefficients $\alpha_{ab}(\omega)$. To that end, let us consider $\lambda(\omega), \mu(\omega)$; these satisfy 
\[w(\etab,\etaa) = 2i\omega\mu(\omega) \text{ and } w(\etab,\overline{\etaa}) = 2i\omega\lambda(\omega).\]
In each case, one proceeds exactly as in \cite{thecauchyproblemforthewaveequationintheschwarzschildgeometry} (and uses the fact $w(\zetab,\zetaa) = w(\etab,\etaa)$) to find $w(\etab,\etaa) = 2i\omega + O(1)$ and $w(\etab,\overline{\etaa}) = O(1)$, which implies 
\[\mu(\omega) = 1 + O\left(\frac{1}{\omega}\right) \text{ and } \lambda(\omega) = O(1)\]
for $\omega$ large. Thus, the coefficients $\alpha_{ab}$ remain at least bounded.

Putting all of this together, we have shown that the integrand in the representation $\eqref{repformula}$ is in $L^1(\R,\C^2)$. Furthermore, this implies that the integral converges pointwise, and thus that the representation $\eqref{repformula}$ holds for each $u\in (-\infty,\uinf)$.\end{proof}

As a simple corollary we now obtain decay:

\begin{cor}The solution $\Psi^{lm}$ of the reduced Cauchy problem $\eqref{cp3}$ vanishes as $t\to \infty$ for fixed $u \in (-\infty,\uinf)$.\label{modaldecay}\end{cor}
\begin{proof}According to the representation formula $\eqref{repformula}$ and the above theorem, $\Psi^{lm}$ is the Fourier transform of an absolutely integrable function. Then by the Riemann-Lebesgue lemma, for fixed $u$, $\Psi^{lm}(t,u) \to 0$ as $t \to \infty$.\end{proof}

Our next goal is to show decay of the solution $\Psi$ of the problem $\eqref{cp2}$. By the uniqueness and convergence of the spherical harmonic decomposition, one obtains a solution of $\eqref{cp2}$ from solutions of $\eqref{cp3}$ via $\eqref{spharmdecomp}$ and vice versa. In particular, this implies that the solution $\Psi$ of the problem $\eqref{cp2}$ has the representation
\begin{equation}\Psi(t,u,\theta,\phi) = \sum_{l=0}^{\infty}\sum_{|m|\leq l} e^{-it\ham}\Psi_0^{lm}(t,u)Y^{lm}(\theta,\phi).\label{repformula2}\end{equation}

Finally, we prove the main theorem of this paper:
\begin{thm}The Cauchy problem $\eqref{cp}$ in the geometry of a generalized Schwarzschild black hole has a unique, smooth, globally defined solution that is compactly supported for all times $t$. Moreover, this solution tends to zero for fixed $(u,\theta,\phi)$ as $t\to \infty$.\label{mainthm}\end{thm}
The proof of this theorem then follows from the modal decay proved in $\eqref{modaldecay}$ and the argument in \cite{thecauchyproblemforthewaveequationintheschwarzschildgeometry}.
\section{Application to the EYM Equations}
\label{eym}
It was shown by Smoller, Wasserman, and Yau in \cite{existenceofblackholesolutionsfortheeymequations} that there exists infinitely many black hole solutions of the SU(2) EYM equations. These solutions correspond to a metric of the form 
\[ds^2 = -T^{-2}(r)dt^2 + A^{-1}(r)dr^2 + r^2d\Omega^2,\]
which has a singularity at some horizon radius $r=r_0>0$ (i.e. $A(r_0) = 0$) and are smooth in the region $(r_0,\infty)$. Moreover, the metric coefficients decay to unity at a rate $O(r^{-1})$ (see section 4 of \cite{existenceofblackholesolutionsfortheeymequations}) and are bounded away from zero away from the singularity. It remains then to analyze the asymptotic behavior near the singularity and the asymptotic decay of the derivatives. To that end, let us state explicitly the differential equations satisfied by $T,A$:
\begin{equation}rA' + (1+ 2w'^2)A = 1 - \frac{(1-w^2)^2}{r^2}, \label{eym1}\end{equation}
\begin{equation}2rA\left(\frac{T'}{T}\right) = \frac{(1-w^2)^2}{r^2} + (1-2w'^2)A - 1.\label{eym2}\end{equation}
There is a third equation allowing one to solve for $w$, but since we only wish to deduce asymptotics, we omit the equation and instead recall the relevant facts about $w$. 
\begin{prop}The function $w$ satisfies the following 
\begin{align}&\lim_{r\to \infty}w^2(r) = 1,\\ &\lim_{r\to \infty} rw'(r) = 0,\\ &\lim_{r\searrow r_0} w^2(r) < 1,\\ &\lim_{r\searrow r_0}|w'(r)| < \infty.\end{align}
Moreover, the following inequality also holds,
\begin{equation}\left(r_0 - \frac{(1-w^2(r_0))^2}{r_0}\right)\neq0.\label{eym3}\end{equation}\end{prop}
For proof, we refer to \cite{existenceofblackholesolutionsfortheeymequations}.
Now, since $A(r)$ is smooth on $[r_0,\infty)$ and $A(r_0) = 0$, a Taylor expansion yields 
\begin{equation}A(r) = A'(r_0) \cdot (r-r_0) + O(r-r_0)^2,\label{Aequation}\end{equation}
where, from $\eqref{eym1}$, we have
\[A'(r_0) = \frac{1}{r_0^2}\left(r_0 - \frac{\left(1-w(r_0)^2\right)^2}{r_0}\right) \neq 0,\]
according to $\eqref{eym3}$. Also, $\eqref{eym2}$ gives
\begin{align}\left(\frac{T'(r)}{T(r)}\right) &= \frac{\left(1-w(r_0)^2\right)^2}{2r_0^3A'(r_0)(r-r_0)}  - \frac{1}{2rA'(r_0)(r-r_0)} +O(1)\notag\\
&=\frac{-1}{2(r-r_0)} + O(1).\label{eym4}\end{align}
Thus, 
\[\frac{d}{dr}\log(T) = \frac{d}{dr}\log\left[(r-r_0)^{-\frac{1}{2}}\right]+O(1),\]
which implies that
\[\frac{d}{dr}\log\left[T\cdot(r-r_0)^{\frac{1}{2}}\right] = O(1).\]
Integrating this from $r_0$ to $r$, we obtain
\[T(r)\cdot(r-r_0)^{\frac{1}{2}} = c_1 + O(r-r_0)\]
for some constant $c_1$, and thus that
\[T(r) = c_1 (r-r_0)^{-\frac{1}{2}} + O(r-r_0)^{\frac{1}{2}}\]
as $r\searrow r_0$. Moreover, using this in $\eqref{eym4}$ yields 
\[T'(r) = c_2(r-r_0)^{-\frac{3}{2}} + O(r-r_0)^{-\frac{1}{2}}\]
for some constant $c_2$. Note also that, in applying the results of this paper to the EYM equations, we make the identification $K^2 = A^{-1}$, and therefore, $K(r) = A^{-\frac{1}{2}}(r)$. Thus from $\eqref{Aequation}$ we find
\[K(r) = c_3 (r-r_0)^{-\frac{1}{2}} + O(1)\]
for $r$ near $r_0$. Finally, we have
\[K'(r) = -\frac{A'(r)}{2A^{\frac{3}{2}}(r)},\]
so we can write
\begin{align*}K'(r) &= -\frac{1}{2}\frac{\left(A'(r_0) + O(r-r_0)\right)}{\left(A'(r_0)(r-r_0) + O(r-r_0)^2\right)^{\frac{3}{2}}} \\
&= c_4(r-r_0)^{-\frac{3}{2}} + O(r-r_0)^{-\frac{1}{2}}\end{align*}
for some constant $c_4$.

Now for the far-field decay condition, observe that from $\eqref{eym1}$, we have
\[A'(r) = -\frac{2(w')^2}{r} + O\left(\frac{1}{r^2}\right),\]
since $A = 1 + O(r^{-1})$. From the relationship between $A$ and $K$, this implies that
\[\left(\frac{K'}{K}\right) = \frac{(w')^2}{r} + O\left(\frac{1}{r^2}\right).\]
Similarly, from $\eqref{eym2}$, we have
\[\left(\frac{T'}{T}\right) = -\frac{(w')^2}{r} + O\left(\frac{1}{r^2}\right).\]
Putting these two observations together yields
\[ \frac{T'}{T} + \frac{K'}{K} = O\left(\frac{1}{r^2}\right)\]
for $r$ tending to infinity.

Thus, black hole solutions of the EYM equations do indeed satisfy the conditions of a generalized Schwarzschild black hole and we conclude that solutions of the Cauchy problem for the wave equation in these geometries must decay according to Theorem~\ref{mainthm}.
\section{Acknowledgements}The author would like to thank his advisor Joel Smoller for introducing me to this problem, for his many helpful discussions, and for his financial support, NSF Contract No. DMS-0603754.
\bibliographystyle{plain}
\bibliography{mybibliography}
\end{document}